\documentclass{article}
\usepackage{amssymb,amsmath,amsthm,mathtools}

\usepackage{fullpage}
\usepackage{graphicx}
\usepackage{subfig}
\usepackage{soul}

\usepackage{tikz,pgfplots}
\pgfplotsset{compat=1.16}

\newtheorem{theorem}{Theorem}
\newtheorem{lemma}[theorem]{Lemma}
\newtheorem{prob}[theorem]{Problem}
\newtheorem{dfn}[theorem]{Definition}

\def\deg{{\rm deg}}
\def\ie{i.e.\ }
\def\etal{{\it et~al.}\,}

\begin{document}

\title{Constructing Order Type Graphs using an Axiomatic Approach}
\author{Sergey Bereg\thanks{Department of Computer Science,
        University of Texas at Dallas.}
\and
Mohammadreza Haghpanah$^*$}
%

\date{}

\maketitle              
\begin{abstract}
A given order type in the plane can be represented by a point set. 
However, it might be difficult to recognize the orientations of some point triples. 
Recently, Aichholzer \etal \cite{abh19} introduced exit graphs for visualizing order types in the plane.
We present a new class of geometric graphs, called {\em OT-graphs}, using abstract order types and their axioms described in the well-known book by Knuth \cite{k92}.  
Each OT-graph corresponds to a unique order type.
We develop efficient algorithms for recognizing OT-graphs and computing a minimal OT-graph for a given order type in the plane. 
We provide experimental results on all order types of up to nine points in the plane including a comparative analysis of exit graphs and OT-graphs.
\end{abstract}
\section{Introduction}
\label{intro}

The {\em orientation} of three noncollinear points in the plane is either clockwise CW or counterclockwise CCW. 
In this paper we assume that point sets are in general position the plane.  
Two finite point sets in the plane have the {\em same order type} if there is a bijection between them preserving orientation of any three distinct points.
The equivalence classes defined by this equivalence relation are the {\em order types} \cite{gp83}.

Recently, Aichholzer \etal \cite{abh19} asked ``... suppose we have discovered an interesting order type, and we would
like to illustrate it in a publication." 
This is exactly the problem that we were facing in our recent paper \cite{bh20} 
where we found that the order type 1874 for 9 points from the database \cite{aak02} 
provides a (tight) lower bound for Tverberg partitions with tolerance 2, see Fig.~\ref{1874}(a). 
Any order type in the plane can be represented by a corresponding point set (or explicit coordinates of the points). 
However, it might be difficult to recognize the orientations of some point triples. 
Aichholzer \etal \cite{abh19} introduced exit graphs for visualizing order types in the plane.
Let $S$ be a set $n$ points in the plane and let $a,b,c\in S$. 
Then $(a,b)$ is an {\em exit edge} with {\em witness} $c$ if there is no $p\in S$ such that 
line $ap$ separates $b$ from $c$ or line $bp$ separates $a$ from $c$, see Fig.~\ref{exit}(a).
Geometrically, it means that an {\em hourglass} defined by $a,b,c$ is empty.
The set of exit edges form the {\em exit graph} of $S$.
To verify that $(a,b)$ is an exit edge with witness $c$, one can check that every point $p\in S\setminus\{a,b,c\}$
is in $A\cup B$, see Fig.~\ref{exit}(a). 
Note that $A\cap S$ and $(B\cap S) \cup\{c\}$ is a partition of $S\setminus\{a,b\}$ by line $ab$. 

\begin{figure}[htb]
\begin{center}
\includegraphics[]{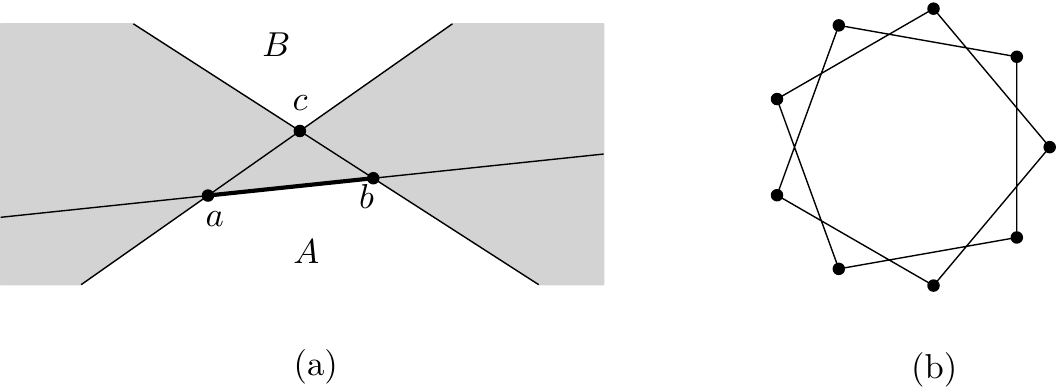}
\end{center}
\caption{(a) Exit edge $(a,b)$ with witness $c$. The hourglass-shaped region (shown in gray) is empty of points.
(b) An exit graph for 9 points in convex position.}
\label{exit}
\end{figure}

We define an {\em OT-graph} on $S$ using two ingredients. 
First, every edge $(a,b)$ in an OT-graph is equipped with the partition of $S$ by line $ab$, i.e.
$S\setminus \{a,b\}=S_{ab}^+ \cup S_{ab}^-$ where $S_{ab}^+$ ($S_{ab}^-$) contains points 
$c\in S$ such that $a,b,c$ has counterclockwise (clockwise) orientation. 
Second, we assume that an OT-graph contains a sufficient number of edges to decide 
the order type of points using axioms described in the well-known book by Knuth \cite{k92}. 
It is easy to visualize the partitions of $S$ for the edges of an OT-graph by drawing lines through them.
This may result in a dense drawing, so we omit lines in the drawing if their partitions can be easily seen.
For example, the OT-graph for the order type 1874 for 9 points from the database \cite{aak02} shown
in Fig.~\ref{1874}(b) has ten edges and only two lines are sufficient.
The property of this graph (since it is an OT-graph) is that the orientation of any triple $abc$
can be decided either (a) directly from the graph if there is an edge with both endpoints in $\{a,b,c\}$,
or (b) algebraically using five axioms \cite{k92}. 

\begin{figure}[htb]
\centering
\subfloat[]{{
\begin{tikzpicture}[scale=0.8]
 \begin{axis}[
  title={Order type 1874},
  xtick=\empty, 
  ytick=\empty, 
  scatter/classes={ 
   a={mark=square*,blue}, 
   b={mark=triangle*,red}, 
   c={mark=*,blue}
  }
 ]
  \addplot[scatter,only marks,
   scatter src=explicit symbolic]
   table[meta=label] { 
   x y label 
    9840 6320  c
    11088 53091  c
    13184 55184  c
    20272 31792  c
    23936 42832  c
    29536 27264  c
    30240 59216  c
    36608 40224  c
    65392 58624  c
   };
  \draw[red, thick] (axis cs: 9840, 6320) -- (axis cs: 11088, 53091);
  \draw[red, thick] (axis cs: 11088, 53091) -- (axis cs: 13184, 55184);
  \draw[red, thick] (axis cs: 13184, 55184) -- (axis cs: 30240, 59216);
  \draw[red, thick] (axis cs: 30240, 59216) -- (axis cs: 65392, 58624);
  \draw[red, thick] (axis cs: 65392, 58624) -- (axis cs: 9840, 6320);
 \end{axis}
\end{tikzpicture}}}\\%
\subfloat[]{{\includegraphics[scale=0.58]{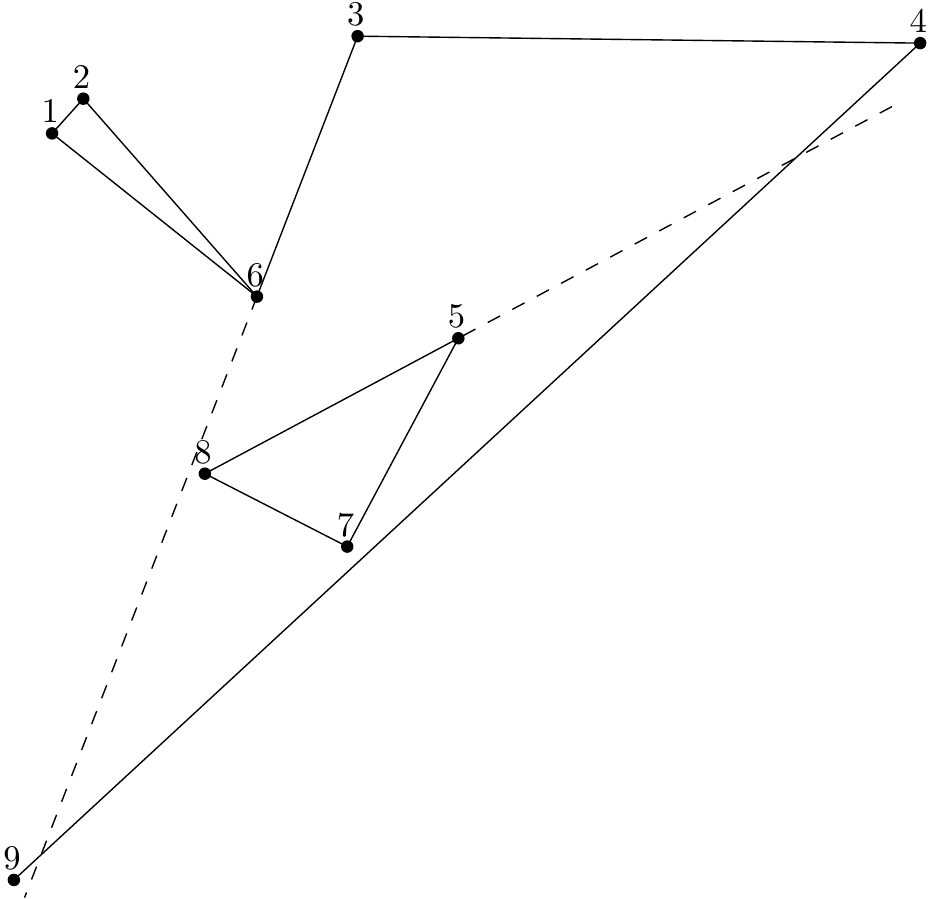}}}
\quad
\subfloat[]{{\includegraphics[scale=0.58]{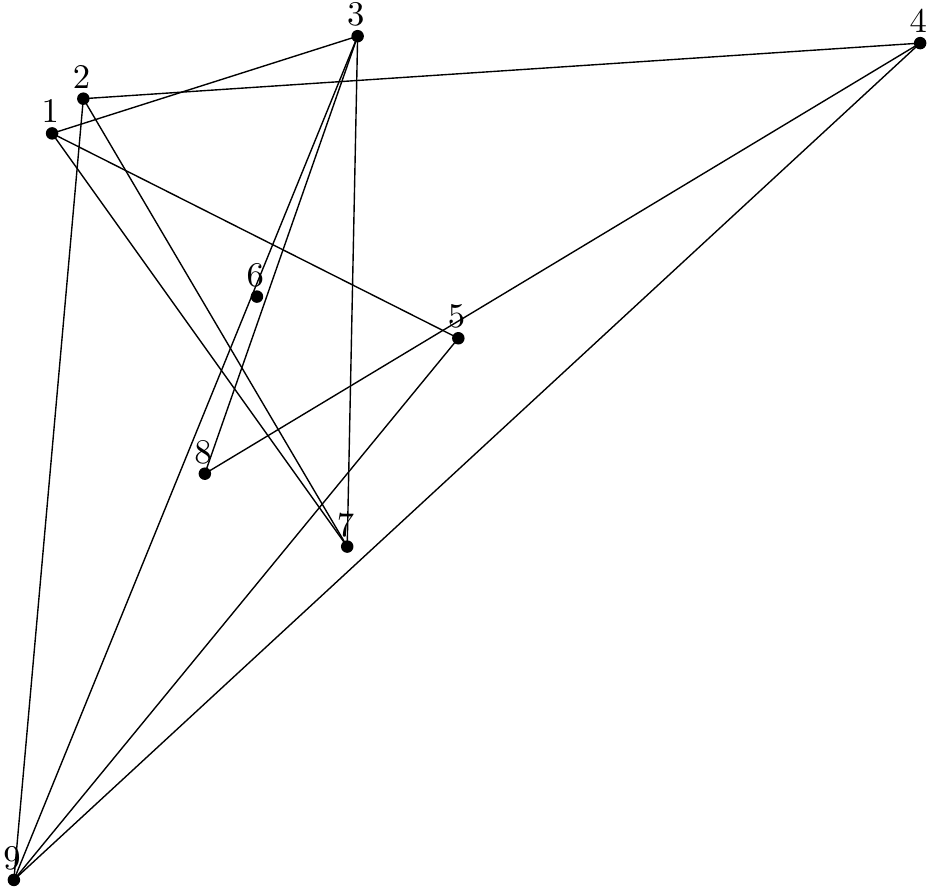}}}
\caption{(a) The order type 1874 for 9 points from the database \cite{aak02}.  
(b) An OT-graph with 9 edges for the order type 1874 (several OT-graphs with 9 edges were computed by an extensive search).  
(c) The exit graph for the order type 1874.}
\label{1874}
\end{figure}

{\em Comparison of OT-graphs and exit graphs}.
Both exit graphs and OT-graphs can be used for visualizing order types of points. It is not sufficient for verifying an order type to just draw such graphs. For exit graphs, one needs to see the witness and the hourglass for every exit edge. For OT-graphs, one needs to see only the lines extending the edges. 
The hourglasses for exit graphs and the lines for OT-graphs are needed only when some triples of points are almost collinear.

Exit graphs and OT-graphs are also different in the following sense.
For a given order type (as a point set), the exit graph is unique but OT-graphs are not since OT-graphs are defined using combinatorial axioms of Knuth \cite{k92}.  
Therefore we have an optimization problem of computing a minimum-size OT-graph for a given order type.
We believe that this optimization problem is NP-hard.
For example, we believe that the OT-graph shown in Fig.~\ref{1874}(b) has the least number of edges (9) for order type 1874 but we do not have a proof for it. 
Note that the OT-graph has 9 edges but the edge graph has 12, see Fig.~\ref{1874}(c) . 

{\em Identification of Order Types}. 
Aichholzer \etal \cite{abh19} suggested requirements for 
a graph representing an order type: "... we want to reduce the number of edges in the drawing as much as possible, but so
that the order type remains uniquely identifiable." 
OT-graphs (including the set of edges and the corresponding partitions) characterize order types, i.e.
each OT-graph corresponds to only one order type.
Unfortunately, it does not hold for the exit graphs. 
As an example, Aichholzer \etal \cite{abh19} constructed 
two sets each of 14 points\footnote{Using a pseudoline arrangement from \cite{fw00}.} such that the exit edges are the same but 
the order types are different.
With respect to minimizing the number of edges, we provide a comparative analysis of exit graphs and OT-graphs of all order types of up to 9 points in Section~\ref{s:experiments}.
Except few cases, OT-graphs have smaller number of edges.
For example, Figure \ref{fig:1268} shows order type 1268 of 9 points
where the exit graph has 15 edges but the OT-graph has only 8 edges.
Furthermore, the OT-graph shown in Fig.~\ref{fig:1268}(b) has non-crossing edges. 

\begin{figure}[h]
\centering
\includegraphics[scale=1.1]{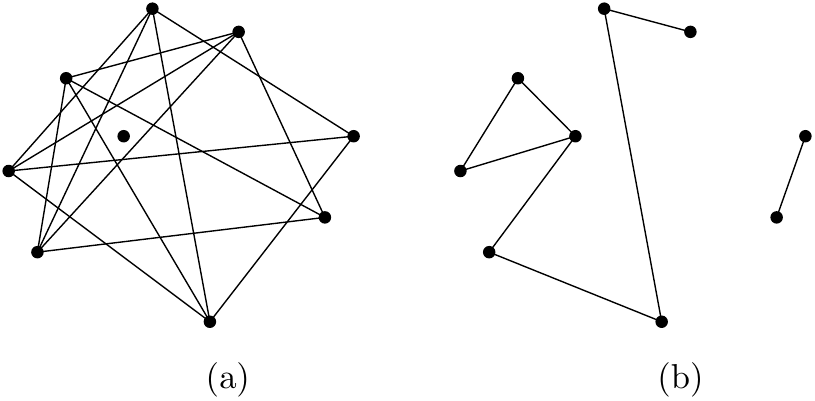}
\caption{The order type 1268 of 9 points represented as (a) the exit graph and (b) the OT-graph.}
\label{fig:1268}
\end{figure}

An interesting question is to find the smallest OT-graphs for points in convex position in the plane.
Let $c_n$ be the minimum number of edges in an OT-graph for $n$ points in convex position.

\begin{theorem} \label{thm:conv}
For any $n\ge 4$, 
$c_{n} \le \lfloor 2n/3\rfloor$. 
\end{theorem}

It is interesting to find exact values of sequence $c_n$.
We experimented with our randomized algorithm from Section~\ref{s:algorithms} and 
conjecture that the bound in Theorem~\ref{thm:conv} is tight for all $n$ up to 20.
It is also interesting that the exit graph for $n$ points in convex position has $n$ edges, see Fig.~\ref{exit}(b) for an example.

{\em Lower bound}. 
Another interesting question is to find the smallest OT-graph for an order type of $n$ points in the plane.
Based on our experiments, it is achieved for points in convex position if $n$ is up to 9.
Is it true for any $n$?
One can argue that $\lceil n/4\rceil$ is a lower bound for the number of edges in any OT-graph 
for $n$ points in convex position. 
It is based on the fact that two consecutive points in the clockwise order along the boundary cannot be 
both isolated in an OT-graph. 

{\em Upper bound}. 
The number of edges in any OT-graph is at most $\binom{n}{2}$.
We prove an upper bound in Section~\ref{s:123} which is smaller than $n^2/4$.
The proof uses the idea of restricting the axioms in OT-graphs.
Specifically, we prove the bound by using only Axioms 1, 2, and 3. 
Surprisingly, in this case, the smallest OT-graphs for any order type of $n$ points 
have the same number of edges depending on $n$ only. 

{\em Algorithms}. 
For any set $T$ of triples with orientations, one can define its {\em CC-closure} 
$Cl(T)$ as the set of all triples that can be derived using Axioms 1-5. 
It is straightforward to make an algorithm for testing in $O(n^5)$ time whether a set of triples $T$ is the closure of itself, i.e. $Cl(T)=T$.
This can be modified to an algorithm for computing the CC-closure for an OT-graph 
(i.e. the set of triples defined by $G$). 
The algorithm repeats the following step.
If new triples are found in the testing algorithm, they are added to the set of triples.
This algorithm has $O(n^8)$ running time.
We show that it can be improved to $O(n^5)$ time.
In Section~\ref{s:experiments} we provide experimental results using our algorithms on all order types of up to nine points in the plane. We also discuss a comparative analysis of exit graphs and OT-graphs using the size of the graphs.
For example, the smallest OT-graphs for all order types for $n=4$ and $n=5$ are shown in Fig.~\ref{fig:n=4}.

\begin{figure}[h]
\centering
\includegraphics{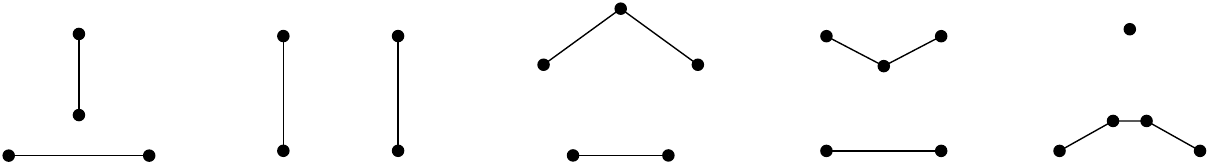}
\caption{Order types for $n=4$ and $n=5$.}
\label{fig:n=4}
\end{figure}


{\em Related Work}.
Order types are studied extensively, see for example the surveys \cite{fg-18,richter1997}. 
Aichholzer \etal \cite{ack16} studied representation of order types using radial orderings.
Cabello \cite{c-pe-06} proved that the problem of deciding whether there is 
a planar straight-line embedding of a graph on a given set of points is NP-complete. 
Goaoc \etal \cite{Goaoc15} explored the application of the theory of limits of dense graphs to order types.
The order types of random point sets were studied in \cite{Cardinal19,devill20}.
Goaoc and Welzl \cite{GoaocW20} studied convex hulls of random order types. 

\section{Preliminaries}
\label{s:pre}

Knuth \cite{k92} introduced and studied {\em CC-systems} (short for "counterclockwise systems") using properties
of order types for up to five points. 
A {\em CC-system} for $n$ points assigns true/false value for every ordered triple of points 
such that they satisfy the following axioms.

{\bf Axiom 1} (cyclic symmetry). $pqr\implies qrp$.

{\bf Axiom 2} (antisymmetry). $pqr \implies \lnot prq$.

{\bf Axiom 3} (nondegeneracy). Either $pqr$ or $prq$.

{\bf Axiom 4} (interiority). $tqr\land ptr\land pqt \implies pqr$.

{\bf Axiom 5} (transitivity). $tsp\land tsq\land tsr\land tpq\land tqr\implies tpr$.

Any set of $n$ points in general position in the plane induces a CC-system 
if we use the "counterclockwise" relation on the points. 
The converse is not true due to the 9-point theorem of Pappus \cite{bok90,k92}. 
When defining a graph for order types using partitions (by the lines extending the edges)
one should be careful. For example, we can ask whether a given set of orientations of some 
triples can be extended somehow to a CC-system. 
If by "extended" we mean finding a CC-system such that the given orientations are preserved in the CC-system, then this problem is NP-complete.
Knuth \cite{k92} proved that it is NP-complete to decide whether specified values 
of fewer than $\binom{n}{3}$ triples can be completed to a CC-system.
We define OT-graphs using the extension of the given orientations by simply applying 
5 axioms. Note that Axioms 1, 2, 4, and 5 imply some orientations.
Axiom 3 also can be formulated as an implication:

{\bf Axiom 3'} (nondegeneracy). $\lnot pqr \implies prq$.

\begin{dfn}
Let $G$ be a graph for a point set $S$ in the plane and let $T$ be the set of triples $abc$ such that $ab, ac$, or $bc$ is an edge of $G$.
Then $G$ is the {\em OT-graph} if the orientation of every triple on $S$ can be derived from $T$ usings Axioms, 1,2,3',4, and 5.
\end{dfn}

\section{Convex Position}
\label{s:conv}

In this section, we explore OT-graphs for point sets in convex position and prove Theorem~\ref{thm:conv}.
Recall that $c_n$ is the minimum number of edges in an OT-graph for $n$ points in convex position.
First, we prove that $c_n\le n$.

\begin{lemma} \label{lem:conv}
Let $S$ be a set of $n$ points in convex position and let $G$ be the graph $(S,E)$ where
$E$ contains the edges of the convex hull of $S$.
Then $G$ is an OT-graph for $S$.
\end{lemma}

\begin{figure}[htb]
\centering
\includegraphics{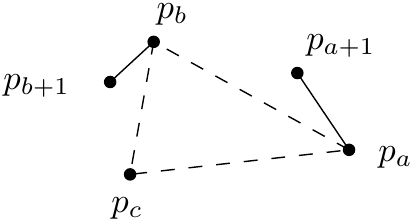}
\caption{Proof of Lemma~\ref{lem:conv}.}
\label{convex-pos1a}
\end{figure}

\begin{proof}
Let $p_0,p_1,\dots,p_{n-1}$ be the points of $S$ in counterclockwise order. 
It suffices to prove that any triple $p_ap_bp_c$ with $0\le a<b<c\le n-1$
has a CCW orientation. 
We prove it by induction on $m=\min\{b-a,c-b,a-c+n\}$.
In the base case, $m=1$. Then $(p_a,p_b), (p_b,p_c),$ or $(p_c,p_a)$ is in $E$. 
Thus, $p_a,p_b,p_c$ has a CCW orientation. 

Suppose that $m>1$ and $m=c-b$. Then $a+1<b$ and $b+1<c$. Edges $(p_a, p_{a+1})$ and $(p_b, p_{b+1})$ imply that triples
$p_a p_{a+1} p_b,  p_a p_{a+1} p_{b+1}, p_a p_{a+1} p_c$, and $p_a p_b p_{b+1}$  
have a CCW orientation. 
By induction hypothesis, triple $p_a p_{b+1} p_c$ has a CCW orientation.
By Axiom 5, $p_a p_b p_c$ has a CCW orientation. 
\end{proof}

\noindent
{\em Proof of Theorem~\ref{thm:conv}}.
Let $p_0,p_1,\dots,p_{n-1}$ be the points of $S$ in counterclockwise order. 
We denote set $\{0,1,\dots,n-1\}$ by $[n]$.

First, suppose that $n=3k$ for some $k\ge 2$.
Consider a graph $G$ with $2k$ edges as shown in Fig.~\ref{convex-pos1}.
We prove that it is an OT-graph.
By Lemma~\ref{lem:conv}, it suffices to show that for any $i,j\in [n]$ with $j\ne i,i+1$ (modulo $n$),
triple $p_ip_{i+1}p_j$ has a CCW orientation\footnote{This condition for a fixed $i$ 
implies that $(p_i,p_{i+1})$ could be an edge in an OT-graph.}.

\begin{figure}[htb]
\centering
\includegraphics{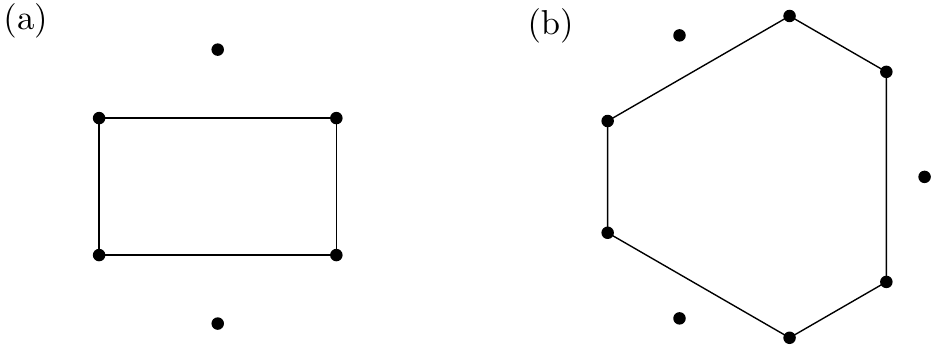}
\caption{OT-graphs for $n$ points in convex position.
(a) $n=6$, (b) $n=9$.}
\label{convex-pos1}
\end{figure}

There are 3 cases to consider, see Fig.~\ref{convex-pos1b}.
Case (a) is clear since $(p_i,p_{i+1})$ is an edge of $G$.
In Case (b), we can assume that $j\ne i+2,i-1$.  
Then it follows by Axiom 5 if we choose $t=p_{i+1}, s=p_{i+2}, p=p_j, q=p_{i-1}$, and $r=p_i$.
In Case (c), we can assume that $j\ne i+2,i-1$.  
Knuth \cite{k92} proved that Axioms 1,2,3, and 5 imply an axiom dual to Axiom 5.

{\bf Axiom 5'} (dual transitivity). $stp\land stq\land str\land tpq\land tqr\implies tpr$.

Then Case (c) follows by Axiom 5' if we choose $t=p_i, s=p_{i-1}, p=p_{i+1}, q=p_{i+2}$, and $r=p_j$.

\begin{figure}[htb]
\centering
\includegraphics{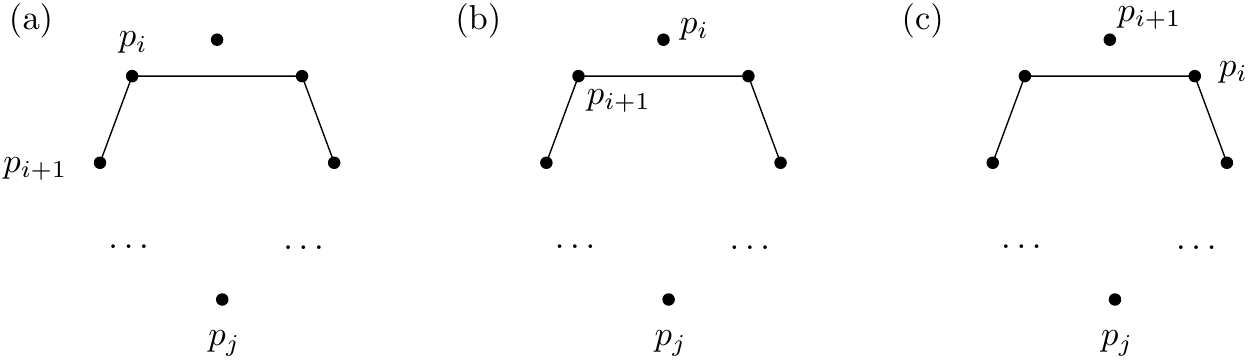}
\caption{Proof of Theorem~\ref{thm:conv} for $n=3k$.}
\label{convex-pos1b}
\end{figure}

Now, suppose that $n=3k+1$ for some $k\ge 2$.
Consider a graph $G$ with $2k$ edges as shown in Fig.~\ref{convex-pos2a}(a).
We prove that it is an OT-graph.
By Lemma~\ref{lem:conv}, it suffices to show that for any $i,j\in [n]$ with $j\ne i,i+1$ (modulo $n$),
tripe $p_ip_{i+1}p_j$ has a CCW orientation. 
If $p_i$ or $p_{i+1}$ is an isolated vertex in $G$ then the argument is the same as in Case (b) and (c)
for $n=3k$, see Fig.~\ref{convex-pos1b}(b) and (c).
If $(p_i,p_{i+1})$ is an edge of $G$ then $p_ip_{i+1}p_j$ has a CCW orientation. 
It remains to consider the case where $(p_i,p_{i+1})$ is one 
of two missing edges in the convex hull at the top, see Fig.~\ref{convex-pos2a}(a).
By symmetry, we assume that $(p_i,p_{i+1})$ is as shown in Fig.~\ref{convex-pos2a}(b).

Suppose that vertex $p_j$ has degree 2 in $G$. 
Let $l$ be the length of path $p_jp_i$ in $G$. 
We show a CCW orientation of $p_ip_{i+1}p_j$ by induction on $l$.
If $l=1$ the orientation follows from edge $p_{i-1}p_i$ of $G$.
If $l>1$ then it follows by Axiom 5 if we choose 
$t=p_{i+1}, s=p_{i+2}, p=p_j, q=p_{j+1}$, and $r=p_i$.
Note that $p_{i+1}p_jq$ has a CCW orientation since $(p_j,q)$ is an edge of $G$.
Also, $p_{i+1}qp_i$ has a CCW orientation by the induction hypothesis.

If vertex $p_j$ is isolated in $G$ then we choose $p,q,r,s,t$ in the same way,
see Fig.~\ref{convex-pos2a}(c). 
Then triple $tpq$ has a CCW orientation from the previous case 
($pq$ is an edge of convex hull). 
And triple $tqr$ has a CCW orientation from the previous case ($q$ has degree 2).
By Axiom 5 tripe $p_ip_{i+1}p_j$ has a CCW orientation. 

\begin{figure}[htb]
\centering
\includegraphics{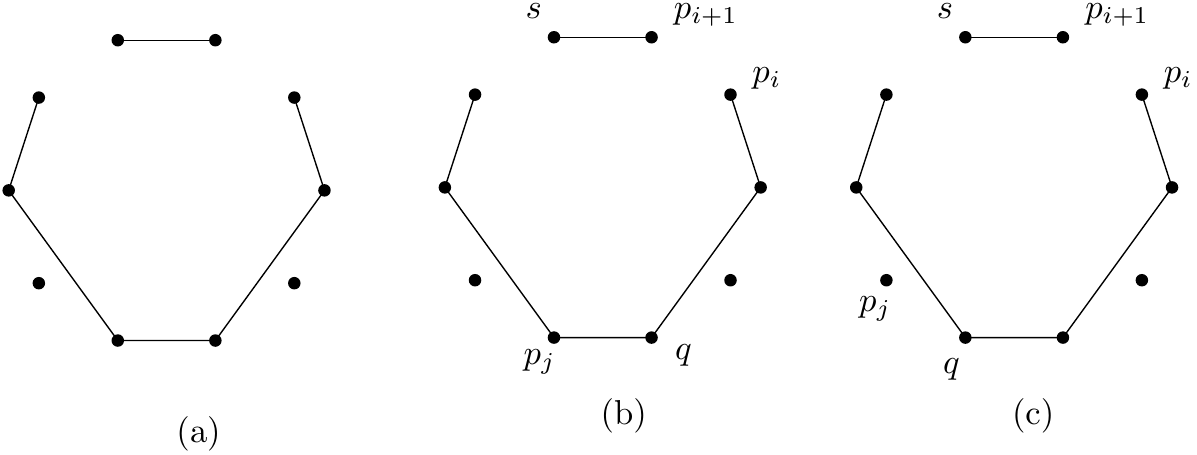}
\caption{Proof of Theorem~\ref{thm:conv} for $n=3k+1$.}
\label{convex-pos2a}
\end{figure}

Finally, suppose that $n=3k+2$ for some $k\ge 2$.
Consider the graph shown in Fig.~\ref{convex-pos3a}.
It is an OT-graph by the same argument as for $n=3k+1$.
\qed

\begin{figure}[htb]
\centering
\includegraphics{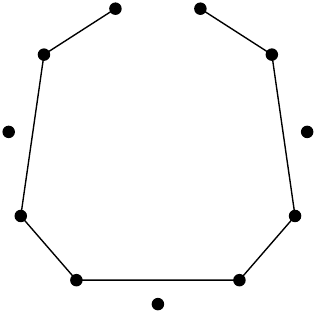}
\caption{OT-graph for $n=3k+2$.}
\label{convex-pos3a}
\end{figure}

{\em Remark}. The OT-graphs presented in the proof of Theorem~\ref{thm:conv}
are not unique. Our program finds also other graphs of the same size, 
see Fig.~\ref{convex-program}.

\begin{figure}[htb]
\centering
\includegraphics{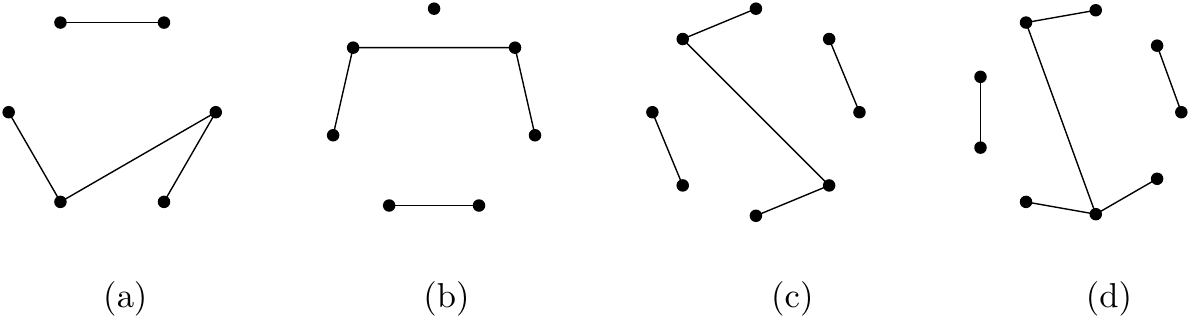}
\caption{OT-graphs for $n=6,7,8,9$ points in convex position computed by a program.}
\label{convex-program}
\end{figure}

\section{Axioms 1,2, and 3 only}
\label{s:123}

Let $e_{123}(S)$ be the minimum number of edges in an OT-graph for a set $S$ of points in general position in the plane if only Axioms 1,2, and 3 are used.
Surprisingly, for any set $S$ of $n$ points (i.e. for any order type), the smallest OT-graph {\em always} contains the same number of edges depending on $n$ only.

\begin{theorem} \label{thm:123}
For any set $S$ of $n\ge 2$ points in general position in the plane, 
$e_{123}(S)=\lfloor \frac{n}{2} \rfloor \lfloor \frac{n - 1}{2} \rfloor$.
\end{theorem}


\begin{proof}
First, we prove it for all even $n$.
Suppose that $n=2k$ for some integer $k\ge 1$.
Let $S$ be a set of $n$ points in general position in the plane and
let $G=(V,E)$ be an OT-graph for $S$ with the minimum number of edges.
We prove that $|E|\ge k^2-k$.

If every vertex in $G$ has degree at least $k$ then $|E|\ge nk/2=k^2>k^2-k$.
Suppose that there is a vertex $v$ of degree at most $k-1$.
Let $N_v\subset V$ be the set of neighbors of $v$ in $G$ and let $U=V\setminus N_v\cup\{v\}$.
Then $|U|=n-1-|N_v|\ge n-k=k$.
If $(u,w)$ is not an edge in $G$ for some $u,w\in U$, then the orientation 
of triple $(v,u,w)$ is not defined by $G$. 
Thus, $U$ is a clique in $G$. 
Let $U'=V\setminus U$, \ie $U'=N_v\cup \{v\}$.
For a vertex $u\in U$, let $\deg_{U'}(u)=|N_u\cap U'|$.
Let $m$ be the smallest degree $\deg_{U'}(u)$ over all $u\in U$.
Then the number of edges $(u,u')\in E$ such that $u\in U$ and $u'\in U'$,
is at least $m|U|$.
Let $u\in U$ be a vertex with $\deg_{U'}(u)=m$.
Then $W$ be the set $U'\setminus N_u$. 
Clearly, $|W|=|U'|-m$ and $W$ is a clique in $G$.
Then $|E|\ge \binom{|U|}{2}+m|U|+\binom{|W|}{2}$. 
Let $x=|W|$. Then $|E|\ge \binom{|U|}{2}+(|U'|-x)|U|+x^2/2-x/2$.
Function $x^2/2-(|U|+1/2)x$  is decreasing in $[0,|U|]$.
Since $x\le |U'|=n-|U|\le k\le |U|$, the minimum of the lower bound of $|E|$ is achieved when $x=|U'|$.
Then $|E|\ge \binom{|U|}{2}+\binom{|U'|}{2}$.
Then $2|E|\ge |U|^2+|U'|^2-|U|-|U'|=|U|^2+|U'|^2-n$.
The minimum of this lower bound is achieved when $|U|=|U'|=k$, i.e. 
$|E|\ge 2\binom{|U|}{2}=k^2-k$. 
This bound is achieved when $|U|=|U'|=k$.
Then any triple $\{a,b,c\}$ has at least two vertices in $U$ or $U'$. 
This implies the theorem for even $n$.

If $n$ is odd, the proof is similar.
If every vertex in $G$ has degree at least $k$ then $|E|\ge nk/2=(k+1/2)k>k^2$.
Again, we assume that there is a vertex $v$ of degree at most $k-1$.
Then $|U|\ge k+1$. The remaining argument is similar to the even case.
Again, $|E|\ge \binom{|U|}{2}+\binom{|U'|}{2}$ where $|U|\ge k+1$ and $|U'|\le k$.
The minimum is achieved when $|U|=k+1$ and $|U'|=k$ (any triple $\{a,b,c\}$ has at least two vertices in $U$ or $U'$).
Then $|E|\ge \binom{k + 1}{2}+\binom{k}{2}=k^2$. The theorem follows.
\end{proof}

\section{Algorithms}
\label{s:algorithms}

Let $G=(S,E)$ be an OT-graph for a set $S$ of $n$ points in the plane.
Let $T(G)$ be the set of triples $abc$ such that $(a,b),(a,c),$ or $(b,c)$ is an edge of $G$.
Note that the orientation of $abc$ is given by the partition of the corresponding edge.
We define the {\em CC-closure} of $G$ as the set all triples that can be proven by applying Axioms 1-5 from $T(G)$.
Note that the {\em CC-closure} can be defined for any subset of triples of points with orientations.

\begin{prob}[{\sc ComputingCC-Closure}]
  \begin{enumerate}
    \item[] {\ }
      \begin{description}
    \item[Given] an OT-graph $G$.
    \item[Compute] the CC-closure of $G$.
  \end{description}
  \end{enumerate}
\end{prob}

A naive approach to solve {\sc ComputingCC-Closure} is to use an algorithm for 
testing CC-closure.

\begin{prob}[{\sc TestingCC-Closure}]
  \begin{enumerate}
    \item[] {\ }
      \begin{description}
    \item[Given] a set of triples with orientations for $n$ points.
    \item[Decide] whether a new triple can be derived using Axioms 1-5.
	If so, find a new triple using Axioms 1-5.
  \end{description}
  \end{enumerate}
\end{prob}

By applying an algorithm for {\sc TestingCC-Closure} to $T(G)$ we can extend $T(G)$ (if possible)
and solve {\sc ComputingCC-Closure}. 
{\sc TestingCC-Closure} can be done in $O(n^5)$ time (since Axiom 5 requires 5 points). 
There are $\binom{n}{3}$ triples and, thus, the naive approach takes
$O(n^8)$ time. 
We show that it can be done much faster.

\begin{theorem}
{\sc ComputingCC-Closure} can be solved in $O(n^5)$ time.
\end{theorem}

\textbf{Algorithm 1}. 

\begin{enumerate}
\item 
Make a list $L_1$ of all input triples with orientations (list $L_1$ stores all triples with known orientations). 
Copy $L_2=L_1$. 
\item 
While list $L_2$ is not empty, remove any triple $abc$ from list $L_2$. 
Apply Axioms as follows. 
Find new triples using Axioms 1,2,3',4, and 5 such that triple $abc$ is used in the condition 
of the axiom with the same orientation.
If a new triple (i.e. not in $L_1$) is found, say $pqr$, then add it to $L_1$ and $L_2$.
\item Return list $L_1$.
\end{enumerate}

\begin{proof}
To implement Algorithm 1 efficiently, we store triples of $L_1$ in 
a 3-dimensional array $A_1$. The value of $A_1[a,b,c]$ is {\em true}/{\em false}
if $abc$ has a CCW/CW orientation; otherwise $A_1[a,b,c]=${\em null}.
Using array $A_1$, we can decide in $O(1)$ time whether a triple is in list $L_1$ or not.
Each triple $abc$ is processed in Step 2 in $O(n^2)$ time since \\
(i) Axioms 1,2, and 3' can be applied at most one time,\\
(ii) Axiom 4 can be applied at most $n-3$ times and\\
(iii) Axiom 5 can be applied at most $(n-3)(n-4)$ times. \\
Each triple is added to (and removed from) list $L_2$ at most one time. 
The number of triples removed from $L_2$ in Step 2 is $O(n^3)$.
Therefore, the total time complexity of the algorithm is $O(n^5)$.

In our implementation of Algorithm 1, we do not maintain list $L_1$. Instead, we compute it in the end using array $A_1$.
\end{proof}

The problem of computing the smallest OT-graph for a given order type
seems complicated. Note that, if we restrict the axioms to Axioms 1,2, and 3
then a simple polynomial-time algorithm for computing the smallest OT-graph exists by Theorem~\ref{thm:123} 
(by constructing two cliques).
Next, we extend Algorithm 1 to a randomized algorithm for computing an OT-graph without 
increasing the running time. 
We incrementally add edges to a graph $G=(S,E)$ until $G$ is an OT-graph for $S$. 
We store $L_1$, a list of triples $abc$ such that $(a,b),(a,c),$ or $(b,c)$ is an edge of $G$. 
Note that the orientation of $abc$ can be computed using the coordinates of $a, b,$ and $c$ in $O(1)$ time. 
As in Algorithm 1, we have list $L_2$ which is useful for computing the CC-closure of $G$.

 \vspace{2mm}
\textbf{Algorithm 2}. 

Input: an order type given by a set of points $S$.

Output: an OT-graph $G$ for $S$
\begin{enumerate}
\item Set $E=\emptyset$. 
Set {\em countCC=0}, the number of triples in the CC-closure of $G=(S,E)$.
\item Compute list $R$ of $\binom{n}{2}$ edges in the complete graph for $S$.
\item Initialize array $A_1[n,n,n]$ with entry values {\em null} and empty list $L_2$.  
\item While {\em countCC}$<n(n-1)(n-2)$ 
\begin{enumerate}
  \item Remove a random edge $(a,b)$ from $R$.
  \item If $A_1[a,b,c]\ne${\em null} for all $c\in S\setminus\{a,b\}$ then continue the "while" loop
  otherwise do the following steps (c) and (d).
  \item Add $(a,b)$ to list $E$. For each $c\in S\setminus\{a,b\}$ such that $A_a[a,b,c]\ne${\em null},
  add one of the triples $(a,b,c)$ or $(b,a,c)$ to list $L_2$ which has a CCW orientation.
\item Process list $L_2$ as in Algorithm 1.
\end{enumerate}
\end{enumerate}

Algorithm 2 (if repeated several times) can find the smallest OT-graph for a given order type,
see for example Fig.~\ref{convex-program}. 
We also make a program that helps to verify the proof of an OT-graph.
Note that a triple can be proven differently using Axioms 1-5.
We develop a program for finding a human-readable proof.
Once the best OT-graph for a given order type is found, the program computes a proof only for triples that 
require Axioms 4 and 5 (Axioms 1-3 are obvious). 
For example, Fig.~\ref{fig:proof} illustrates an OT-graph among all order type of 9 points and the format of the proof. 

\begin{figure}[h]
\centering
\includegraphics{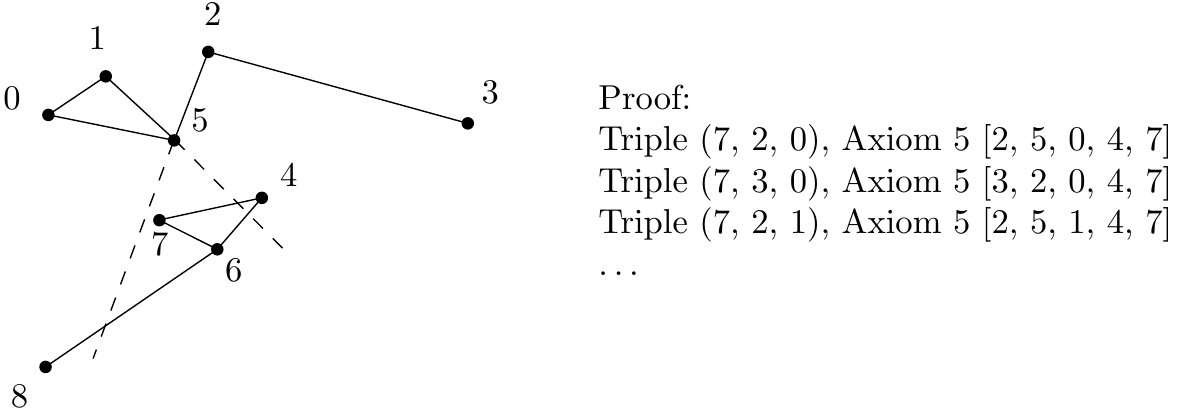}
\caption{An OT-graph for order type of 9 points and a part of the proof of it. 
The format for Axiom 4 is [$p,q,r,t$]
and the format for Axiom 5 is [$p,q,r,s,t$].}
\label{fig:proof}
\end{figure}

{\em Greedy algorithm}.
Each iteration Algorithm 2 is quite fast (for relatively small $n$).
However, it may require many runs to find small OT-graphs.
Another possibility is a greedy algorithm where all possible edges for adding to the current graph are tested and the edge maximizing the size of the CC-closure is selected. 
Since the computation of the CC-closure takes $O(n^5)$ time, this approach is computationally expensive (it takes $O(n^7)$ time for selecting one edge and $O(n^9)$ for constructing the OT-graph).
We developed a different greedy algorithm where the edge maximizing the size of the CC-closure using only Axioms 1,2,3 is selected.
We found an implementation of this algorithm without increasing the running time, i.e. with running time $O(n^5)$.
We add a new 2-dimensional array $C[..]$ for counting triples corresponding to the edges.
Initially, $C[a,b]=n-2$ for all pairs $(a,b)$ of points $a\ne b$.
Every time a new triple, say $abc$, is proven using Axioms we subtract one from $C[x,y]$ for all possible  $x\ne y\in \{a,b,c\}$. 
Then, the greedy selection can be done by finding an edge $(a,b)$ maximizing $C[a,b]$.

The total running time of this algorithm has two components.
It is $O(n^5)$ time as in the randomized algorithm plus the total time for processing new array $C[..]$.
There are $O(n^3)$ new triples and each triple requires $O(1)$ to update $C[..]$. 
This step takes $O(n^3)$ time in total.
The computation of a new edge for $G$ takes $O(n^2)$ time. 
Thus, the total time for computing the edges of $G=(S,E)$ is $O(mn^2)$ where $m=|E|$.
Therefore, the total time for processing array $C[..]$ is $O(n^4)$. 
{\em Minimal OT-graphs}.
When an OT-graph with $m$ edges is computed, it can be checked for minimality.
An OT-graph for some order type is {\em minimal} if removal of any edge results in a graph 
which not an OT-graph, i.e. its CC-closure does not contain all possible triples.
This can be decided by applying the algorithm for {\sc ComputingCC-Closure} $m$ times.

\section{Experiments}
\label{s:experiments}

We implemented the randomized algorithm (Algorithm 2) and the greedy algorithm for computing OT-graphs. 
The programs are written in Java 8 using multi-threading and thread synchronization. 
We used a Linux server with 32 CPUs and 62 GB RAM to execute our program.
We have computed the exit graphs and the OT-graphs on the database of order types \cite{aak02} for $n=3,4,\dots,9$.
To achieve current database and ensure the minimality of edges of OT-graphs, we run it around more than 3 days on the dataset.
The results are shown in Table~\ref{tab1}.
Our experiments show that in many cases the greedy algorithm outperforms Algorithm 2 by the size of an OT-graph. 
Therefore, we iterate the greedy algorithm (with random tie-breaking) first and then iterate Algorithm 2 searching for a possible improvement. 
The number of iterations used for the greedy algorithm was 300,1200,10000 for $n=7,8,9$ respectively. 
The number of iterations used for Algorithm 2 significantly larger (100000 for $n=9$).
About 70\% of OT-graphs in Table~\ref{tab1} were computed using the greedy algorithm. 
The improvement achieved by Algorithm 2 was rather small:
typically one edge reduction for an order type. 
The program implementing Algorithm 2 is still running and hopefully, new OT-graphs will be computed in a few months.

\begin{table}
\caption{OT-graphs for $n$ up to 9. Column $i, i=1,2,\dots,11$ contains the number of OT-graphs with $i$ edges.}
\label{tab1}
\begin{center}
{\setlength{\tabcolsep}{1.6mm}
 \begin{tabular}{|c | r r r r r r r r r r r | r |} 
 \hline
 $n$ & 1 & 2 & 3 & 4 & 5 & 6 & 7 & 8 & 9 & 10 & 11 & total \\ [0.5ex] 
 \hline\hline
 3 & 1 & & & & & & & & & & & 1 \\
 \hline
 4 & & 2 & & & & & & & & & & 2 \\
 \hline
 5 & & & 3 & & & & & & & & & 3 \\
 \hline
 6 & & & & 14 & 2 & & & & & & & 16 \\
 \hline
 7 & & & & 2 & 79 & 54 & & & & & & 135 \\
 \hline
 8 & & & & & 26 & 696 & 1,802 & 791 & & & & 3,315 \\
  \hline
 9 & & & & & 1 & 234 & 9,379 & 49,331 & 73,906 & 25,671 & 295 & 158,817 \\
 \hline
\end{tabular}
}
\end{center}
\end{table}

\begin{figure}[h]
\centering
\includegraphics{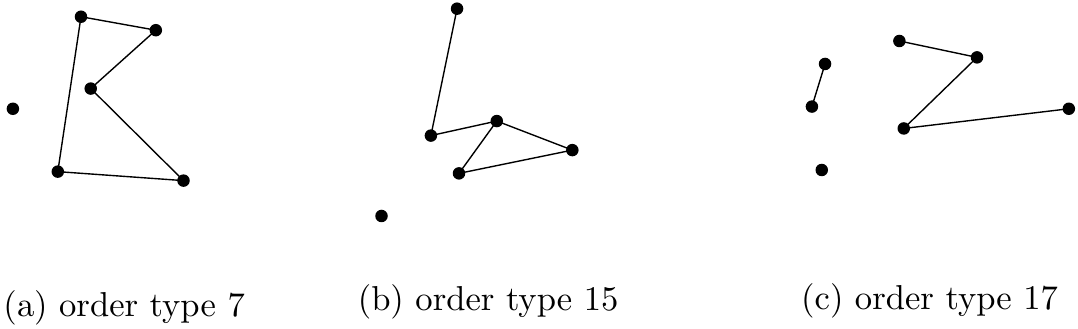}
\caption{Extreme OT-graphs for $n=6$ and $n=7$. 
(a),(b) Two order types for $n=6$ maximizing the number of edges. (c) An order type for $n=7$ (different from the convex case) minimizing the number of edges.}
\label{fig:6}
\end{figure}

It is interesting that the smallest OT-graphs are achieved for 14 order types with $n=6$, 
for 2 order types with $n=7$, for 26 order types with $n=8$, and for 124 order types with $n=9$.  
There are only 2 order types for $n=6$ whose OT-graphs require 5 edges.
They are shown in Fig.~\ref{fig:6}(a) and (b).
The two order types for $n=7$ that admit OT-graph with 4 edges are shown in Fig.~\ref{convex-program}(b)
(the convex position) and in Fig.~\ref{fig:6}(c).

Let $\mu(n)$ be the minimum number of edges in an OT-graph for $n$ points. 
Based on our experiments, 
we conjecture that $\mu(4)=2$,
$\mu(5)=3, \mu(6)=\mu(7)=4, \mu(8)=\mu(9)=5$. 
This can be compared with exit graphs where the minimum number of edges is the same for $n=5,6,7,8$ but is larger for $n=9$, see Fig.~\ref{fig:comp}.

Figure~\ref{fig:comp}(a),(c) shows the distribution of the graph sizes (OT-graph vs exit graph). 
Figure~\ref{fig:comp}(b),(d) shows comparison of the graph sizes for each order type (the order types are sorted by the size of OT-graph and exit graph). 
Except one order type for $n=8$ and 17 order types for $n=9$, the OT-graphs are smaller that the exit graphs. 
For $n=9$, the maximum size of OT-graph/exit graph is 11/16, respectively.
The corresponding total number of edges is 
1386819 for OT graphs and 1673757 for exit graphs which is 82.85\%.

\begin{figure}[htb]
\subfloat[]{\includegraphics[scale=0.46]{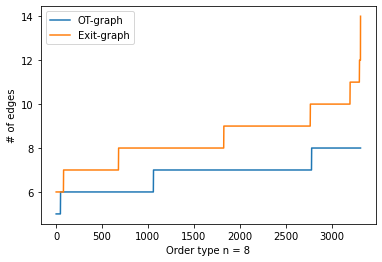}}
\subfloat[]{\includegraphics[scale=0.46]{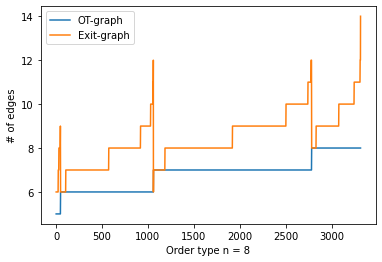}}
\\
\subfloat[]{\includegraphics[scale=0.46]{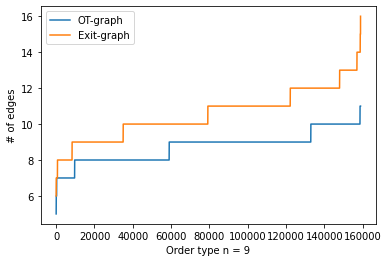}}
\subfloat[]{\includegraphics[scale=0.46]{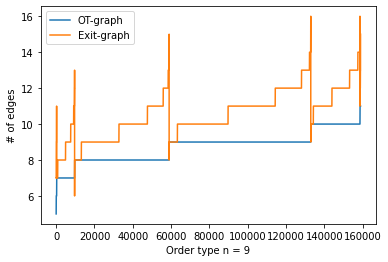}}
\caption{The sizes of OT-graphs and exit graphs of order types for (a,b) $n=8$ and (c,d) $n=9$. The order types in (a) and (c) are sorted independently for OT-graphs and exit graphs. The functions  in (b) and (d) use the same order type on the $x$-axis (the order types are sorted lexicographically).}
\label{fig:comp}
\end{figure}

\section{Concluding Remarks}
\label{s:concl}

In this paper, we introduced OT-graphs for visualizing order types in the plane.
This new concept gives rise to many interesting questions.
Is it true that the smallest size OT-graphs for all order types of $n$ points 
are achieved for points in convex position? 
Is the bound in Theorem~\ref{thm:conv} tight? 

In many cases there are different OT-graphs of minimum size for the same order type.
One can use other criteria to optimize OT-graphs, for example, crossings.
Fig.~\ref{fig:n=4} shows that there exist OT-graphs without crossings for all order types of 4 and 5 points.
Theorem~\ref{thm:conv} shows that there are OT-graphs without crossings for points in convex position. 
Can it be generalized in this sense?


\bibliographystyle{plain}
\bibliography{main}

\begin{thebibliography}{10}

\bibitem{aak02}
Oswin Aichholzer, Franz Aurenhammer, and Hannes Krasser.
\newblock Enumerating order types for small point sets with applications.
\newblock {\em Order}, 19(3):265--281, 2002.

\bibitem{abh19}
Oswin Aichholzer, Martin Balko, Michael Hoffmann, Jan Kyn\v{c}l, Wolfgang
  Mulzer, Irene Parada, Alexander Pilz, Manfred Scheucher, Pavel Valtr, Birgit
  Vogtenhuber, and Emo Welzl.
\newblock Minimal representations of order types by geometric graphs.
\newblock In {\em Graph Drawing (Proc. GD '19)}, pages 101--113, 2019.

\bibitem{ack16}
Oswin Aichholzer, Jean Cardinal, Vincent Kusters, Stefan Langerman, and Pavel
  Valtr.
\newblock Reconstructing point set order types from radial orderings.
\newblock {\em Internat. J. of Comput. Geom. \& Applications},
  26(03--04):167--184, 2016.

\bibitem{bh20}
Sergey Bereg and Mohammadreza Haghpanah.
\newblock New lower bounds for {T}verberg partitions with tolerance in the
  plane.
\newblock {\em Discrete Applied Mathematics}, 283:596 -- 603, 2020.

\bibitem{bok90}
J{\"u}rgen Bokowski, J{\"u}rgen Richter, and Bernd Sturmfels.
\newblock Nonrealizability proofs in computational geometry.
\newblock {\em Discrete \& Computational Geometry}, 5(4):333--350, 1990.

\bibitem{c-pe-06}
Sergio Cabello.
\newblock Planar embeddability of the vertices of a graph using a fixed point
  set is {NP}-hard.
\newblock {\em J. Graph Algorithms Appl.}, 10(2):353--363, 2006.

\bibitem{Cardinal19}
Jean Cardinal, Ruy~Fabila Monroy, and Carlos Hidalgo{-}Toscano.
\newblock Chirotopes of random points in space are realizable on a small
  integer grid.
\newblock In Zachary Friggstad and Jean{-}Lou~De Carufel, editors, {\em
  Proceedings of the 31st Canadian Conference on Computational Geometry, {CCCG}
  2019, August 8-10, 2019, University of Alberta, Edmonton, Alberta, Canada},
  pages 44--48, 2019.

\bibitem{devill20}
Olivier Devillers, Philippe Duchon, Marc Glisse, and Xavier Goaoc.
\newblock On order types of random point sets.
\newblock {\em CoRR}, abs/1812.08525, 2018.

\bibitem{fg-18}
Stefan Felsner and Jacob~E Goodman.
\newblock Pseudoline arrangements.
\newblock In {\em Handbook of Discrete and Computational Geometry}, pages
  125--157. Chapman and Hall/CRC, 2017.

\bibitem{fw00}
Stefan Felsner and Helmut Weil.
\newblock A theorem on higher {B}ruhat orders.
\newblock {\em Discret. Comput. Geom.}, 23(1):121--127, 2000.

\bibitem{Goaoc15}
Xavier Goaoc, Alfredo Hubard, R{\'{e}}mi de~Joannis~de Verclos,
  Jean{-}S{\'{e}}bastien Sereni, and Jan Volec.
\newblock Limits of order types.
\newblock In Lars Arge and J{\'{a}}nos Pach, editors, {\em 31st International
  Symposium on Computational Geometry, SoCG 2015, June 22-25, 2015, Eindhoven,
  The Netherlands}, volume~34 of {\em LIPIcs}, pages 300--314. Schloss Dagstuhl
  - Leibniz-Zentrum f{\"{u}}r Informatik, 2015.

\bibitem{GoaocW20}
Xavier Goaoc and Emo Welzl.
\newblock Convex hulls of random order types.
\newblock In Sergio Cabello and Danny~Z. Chen, editors, {\em 36th International
  Symposium on Computational Geometry, SoCG 2020, June 23-26, 2020,
  Z{\"{u}}rich, Switzerland}, volume 164 of {\em LIPIcs}, pages 49:1--49:15.
  Schloss Dagstuhl - Leibniz-Zentrum f{\"{u}}r Informatik, 2020.

\bibitem{gp83}
Jacob~E Goodman and Richard Pollack.
\newblock Multidimensional sorting.
\newblock {\em SIAM Journal on Computing}, 12(3):484--507, 1983.

\bibitem{k92}
Donald~E. Knuth.
\newblock {\em Axioms and Hulls}, volume 606 of {\em Lecture Notes in Computer
  Science}.
\newblock Springer, 1992.

\bibitem{richter1997}
J{\"u}rgen Richter-Gebert and G{\"u}nter~M Ziegler.
\newblock Oriented matroids.
\newblock In {\em Handbook of discrete and computational geometry}, pages
  159--184. Chapman and Hall/CRC, 2017.

\end{thebibliography}

\end{document}